\documentclass[5p]{elsarticle}
\usepackage{lineno,hyperref}
\modulolinenumbers[5]
\journal{Sustainable Energy, Grids and Networks.}
\usepackage{microtype}
\usepackage{graphicx}
\usepackage{flafter}
\usepackage{todonotes}
\usepackage{soul}
\usepackage{amsmath,amssymb}
\usepackage{amsthm}
\newtheorem{theorem}{Theorem}
\newtheorem{definition}{Definition}
\usepackage{siunitx}
\usepackage{wasysym}
\usepackage{hyperref}
\usepackage[noabbrev,capitalise]{cleveref}
\usepackage{aliascnt}
\newaliascnt{lemma}{theorem}
\newaliascnt{corollary}{theorem}
\newtheorem{corollary}[corollary]{Corollary}
\newtheorem{lemma}[lemma]{Lemma}
\aliascntresetthe{lemma}
\aliascntresetthe{corollary}
\crefname{lemma}{Lemma}{Lemmas}
\crefname{corollary}{Corollary}{Corollaries}
\crefname{subsection}{Subsection}{Subsections}
\graphicspath{{./}}

\DeclareMathOperator*{\argmin}{arg\,min}

\begin{document}
\begin{frontmatter}
    \title{
        Distributed Coordination of Deferrable Loads: A Real-time Market with Self-fulfilling Forecasts\tnoteref{funding}
    }
    \tnotetext[funding]{
        This project has received funding from the European Union's Horizon 2020 research and innovation programme under Marie Sklodowska-Curie grant agreement No.\ 675318 (INCITE).
    }
    
    \author[ESE,aast]{Hazem A. Abdelghany}
    \address[ESE]{Department of Electrical Sustainable Energy, Faculty of Electrical Engineering, Mathematics and Computer Science, Delft University of Technology, The Netherlands.}
    \cortext[mycorrespondingauthor]{Corresponding author}
    \ead{h.a.m.f.abdelghany@tudelft.nl}

    \author[ESE]{Simon H. Tindemans}
    
    \author[ST]{Mathijs M. de Weerdt}
    
    \address[ST]{Department of Software Technology, Faculty of Electrical Engineering, Mathematics and Computer Science, Delft University of Technology, The Netherlands.}
    
    \author[ESE,cwi]{Han la Poutr\'e}
\address[cwi]{Centrum Wiskunde \& Informatica (CWI), Amsterdam, The Netherlands.}
    \address[aast]{Electrical and Control Engineering Department, Arab Academy for Science, Technology and Maritime Transport, Cairo, Egypt.}

    \begin{abstract}
        Increased uptake of variable renewable generation and further electrification of energy demand necessitate efficient coordination of flexible demand resources to make most efficient use of power system assets. Flexible electrical loads are typically small, numerous, heterogeneous and owned by self-interested agents. Considering the multi-temporal nature of flexibility and the uncertainty involved, scheduling them is a complex task. This paper proposes a forecast-mediated real-time market-based control approach (F-MBC) for cost minimizing coordination of uninterruptible time-shiftable (i.e.\ deferrable) loads. F-MBC is scalable, privacy preserving, and usable by device agents with small computational power. Moreover, F-MBC is proven to overcome the challenge of mutually conflicting decisions from equivalent devices. Simulations in a simplified but challenging case study show that F-MBC produces near-optimal behaviour over multiple time-steps.
    \end{abstract}


\end{frontmatter}

 \section{Introduction}
    Power systems have seen an increasing penetration of distributed energy resources (DERs), such as distributed generators, flexible demand, and small-scale renewable generation. This trend has significant impacts on the network, leading to congestion, reduced network utilization, and even instability or system inoperability at the distribution level~\cite{FutureGrid}. Consequently, the transition to future power systems requires either a great deal of investment in grid reinforcement, or efficient use of flexibility from DERs through coordination.
    
    Optimal coordination among DERs is a complex multi-dimensional problem, especially in settings with small, numerous, heterogeneous DERs owned by self-interested agents. The complexity is further amplified by inter-temporal constraints introduced by shifting energy consumption and uncertainties in DER usage patterns and renewable-based generation. A suitable coordination approach for such a setting is required to be simple and usable by agents with small computational power~\cite{small_comp}, scalable for settings with numerous DERs, and privacy preserving since the DERs being considered are owned by self-interested agents.
    
    This problem has been considered in a number of settings, including electric vehicle charging~\cite{vanderLinden2018optimal}, deferrable loads such as washing machines, dish washers, and thermostatically controlled loads.  Most control techniques for flexible demand are based either on centralized coordination, top-down control, or price response~\cite{flexreview,state}. Centralized and top-down approaches (e.g.~\cite{topdown1}) are not suitable when considering privacy, autonomy, and scalability constraints, whereas completely decentralized approaches relying on one-way communication (e.g.\ price response~\cite{rtp2,rtp}) have uncertain realized system response. A comprehensive review of advantages and disadvantages of control approaches can be found in~\cite{powermatcher}.
    
    \subsection{Market-based Control}
        A natural fit for the problem of coordinating self-interested DERs is transactive control, which refers to control approaches which perform coordination and control tasks by using economic incentive signaling to exchange information about generation, consumption, constraints, and responsiveness of assets over dynamic, real-time forecasting periods~\cite{transactive}.
        Market-based control (MBC) describes a class of transactive control algorithms that take the form of a mediated market~\cite{MBC}. In an attempt to find a middle way between the aforementioned approaches, this paradigm provides simultaneously a degree of privacy, autonomy, certainty, and openness compared to the aforementioned approaches~\cite{state, powermatcher}. However, when used for coordination among numerous DERs, over multiple time-steps and taking into account uncertainty, MBC approaches rapidly grow in complexity, limiting their scalability and practical feasibility. For example, multi-settlement markets, such as in~\cite{multi1, multi2} require complex bid formulation algorithms, which is especially hard for devices with small computational power~\cite{small_comp}. Accounting for uncertainty similarly increases complexity, as is evident in the hierarchical MBC approach in~\cite{complex}. In~\cite{iterative1, iterative2}, iterative approaches for coordination were proposed. An iterative approach based on Mean-field games was proposed in~\cite{MFG}. However,~\cite{fast_bidding} indicates iterative approaches are not suitable for real-time operations due to uncertain convergence time and dependence on initial conditions. The same logic applies for negotiation approaches such as in~\cite{neg1}. On the other hand, approaches based on the assumption of cooperative agents~\cite{coop1,coop2, cooperative} are not suitable for the settings with self-interested agents. 
        
    \subsection{Real-time Market-based Control}
        In this paper, we use the term ``Real-time market-based control (RTMBC)'' to describe a simple and scalable form of MBC. In RTMBC, DERs are represented by autonomous agents participating in a spot power market. The market is cleared for the upcoming time-step (i.e.\ in real time) by means of a double auction. The use of decentralized decision making and a centralized one-shot market clearing simplifies the whole process. Device level constraints and objectives are taken into account in the process of bid/offer formulation. An example of such approach can be seen in~\cite{powermatcher}. 
            
        Despite these beneficial properties, in practice RTMBC often leads to poor performance over multiple time-steps due to uncertainty, inter-temporal constraints of uninterruptible devices, and mutually-conflicting decisions that arise from decentralization and the self-interested behaviour of agents~\cite{mutually,heterogenity_instability,fast_bidding}. For example, in~\cite{ancillary} the effect of such behaviour is shown to lead to exhaustion of flexibility in the system. An approach for coordination among thermostatically controlled loads was presented in~\cite{TCL}. This was further studied in~\cite{mutually} where it was found prone to load synchronization and power oscillations. Agents submitting similar bids (i.e.\ Bulk switching), and clustering at lower price periods are phenomena that occur when optimal decisions from the agents' perspective conflict and lead to sub-optimal outcomes both at the agent level and system level. This is most apparent in case of identical devices given the same information. Therefore, identical devices pose a challenge to many coordination approaches.
            
    \subsection{Summary of contributions}
        In this paper, we aim at solving the problem of scheduling a set of uninterruptible deferrable loads over multiple time-steps to minimize generation cost taking into account uncertainty. We will refer to this as the ``optimal coordination problem''. To achieve this, we propose the forecast-mediated market-based control approach (F-MBC). F-MBC relies on decentralized bid formulation and centralized one-shot market clearing to coordinate among these devices. The proposed approach is scalable and preserves end-user privacy and autonomy. It relies on probabilistic price forecasts obtained by a facilitator that accounts for uncertainty in renewable-based generation and DER usage patterns. Moreover, we design a low-complexity Markov decision process(MDP) based optimal bidding algorithm for deferrable loads, which is usable by a device with limited computational power (e.g.\ embedded systems) to formulate a bid that minimizes its own expected cost given probabilistic price forecasts. We show that the combination of probabilistic reference prices, optimal bidding, and real-time market clearing solves the problem of mutually-conflicting decisions among identical devices; that is, two  identical device agents with different deadlines will never have the same bid. This is shown mathematically in \cref{sec: method}. Additionally, we design a tie-breaking mechanism to assist in market clearing when several agents are indifferent between different actions at the market-clearing price. Moreover, we prove approximate consistency of the approach by bounding the deviation from the optimal solution that occurs if the forecast correctly identifies an optimal feasible solution. We show by simulation that the proposed F-MBC approach achieves near-optimal system level performance over multiple time-steps (i.e.\ minimizes overall generation cost) in \cref{sec: result}.
            
\section{Methodology} \label{sec: method}
     Consider a setting of uninterruptible deferrable loads, with deadlines set by their respective owners. This resembles a collection of devices such as irrigation pumps, greenhouse lighting, or home appliances such as washing machines, dryers, etc.~\cite{example1,example2}. We assume that each of the deferrable loads acts in its economic best interest, minimizing its consumption cost subject to device level constraints (e.g.\ deadline, uninterruptibility). 
     
     The challenge is to design a scheme that fully or approximately solves the optimal coordination problem, scheduling the flexible demand over multiple discrete time-steps with the objective of minimizing the overall generation cost. It is important to note that the global cost minimization is equivalent to social welfare maximization  since the total energy demand (and, therefore, the utility) is fixed. Therefore, for the remainder of the paper we will just use the term ``optimal coordination''. 
     
     To achieve this, we rely on the idea of ``self-fulfilling forecasts''. As illustrated in \cref{fig:overview}, F-MBC comprises three types of autonomous agents; A facilitator, an auctioneer, and a device agent per flexible device. The facilitator is a central entity which, in general, does not have access to private information (e.g.\ deadlines, cycle durations) and cannot directly control the devices. This is a sensible assumption in settings where DERs are small, numerous and owned by self-interested agents. Such an approach is similar to the vision of layered decentralized optimization architecture in~\cite{Layered}. The facilitator utilizes aggregate historical information, forecasts, behaviour patterns, and system models to estimate an ``offline optimal'' solution to the optimal coordination problem. Some examples of techniques to solve such a problem can be found in~\cite{SO,MCPforecast,forecasts}. The resulting estimated schedule is probabilistic and results in a probabilistic reference price for each time-step (in the form of a probability distribution), thus taking into account uncertainty. Throughout, we assume that the price of energy paid by devices equals the marginal cost of generation at the relevant time step.
     The probabilistic reference prices are then communicated to the flexible demand agents which use this information for bid formulation. Device agents formulate their respective bids in a self-interested manner (i.e.\ minimizing the expected cost incurred by the agent). A device agent takes into account local deadline and uninterruptibility constraints in addition to the probabilistic reference prices provided by the facilitator. Bids are then submitted to a central auctioneer in the form of a demand function. Finally, an allocation is made through a one-shot double auction and an additional tie breaking mechanism. The facilitator updates the ``estimate'' for the future taking into account the market outcome which results in an updated probabilistic reference price signal. The whole process is repeated for every time-step.

     It is noteworthy here that aggregation of bids can be done centrally or through hierarchical aggregation of bid functions. This means that the complexity of aggregating bids is linear, at worst, or logarithmic, at best, when the system is organized as a binary tree. This, combined with decentralized bid optimization, one-shot market clearing and the non-iterative nature of the approach makes it scalable and simple to implement even in scenarios where agents have small computational power. Moreover, the outcome of this process is a near-optimal system-level behaviour over multiple time-steps. The resulting coordination approximates the ``offline optimal coordination'' estimated a priori, so the probabilistic reference prices can be considered ``self-fulfilling''.

    \begin{figure}
    		\centering
    		\includegraphics[width=0.8\linewidth,keepaspectratio]{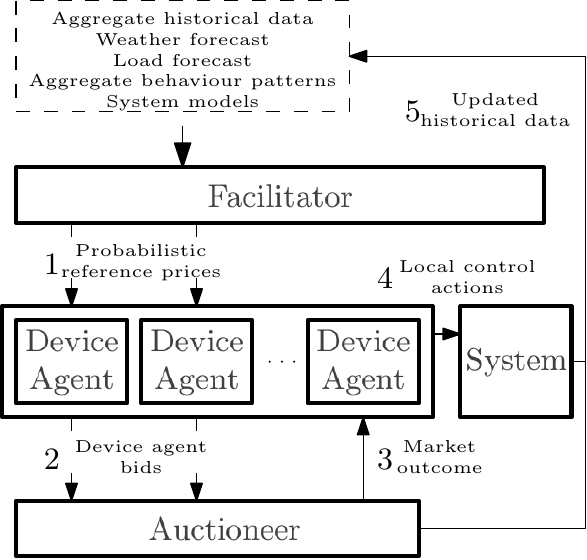}
    		\caption{Schematic overview of the proposed F-MBC approach}
    		\label{fig:overview}
	\end{figure}
    
    \subsection{Mathematical Framework}
        Consider a scheduling horizon consisting of the set of discrete time steps $\mathcal{T}=\{1,\ldots, \text{T}\}$ with fixed intervals $\Delta t$. The subscript $t$ will be used to refer both to the instant $t$ as well as the interval that immediately follows, depending on the context. The system comprises a set $\mathcal{A}$ of uninterruptible deferrable devices owned by self-interested consumers. Each device is represented by an agent $a$ defined by a deadline, duration and a power consumption pattern $d^{a},\,D^{a},\,\{P^a_0, \ldots ,P^a_{D^{a}-1}\}$ respectively. The system also has inflexible demand, and flexible generation with a non-decreasing marginal cost $m_t(P)$ (which may include zero-cost renewable generation). An optimal coordination denotes the allocation of flexible devices over the scheduling horizon, such that the overall cost of generation to meet the aggregate demand $P_{1:T}$ is minimized: 
        \begin{equation}
            P^*_{1:T} = \argmin_{P_{1:T}}  \sum_{i=1}^{\text{T}} \Delta t \int_{0}^{P_{t}} m_{t}(P') \mathrm{d}P'.
        \end{equation} 
        This is subject to system-level constraints (i.e.\ supply/demand matching, flexible generation limits), and agent-level constraints (i.e.\ deadlines, uninterruptibility).
        
        \subsection{MDP-based Optimal Bidding}
            At this point, we describe how an agent may compute and optimize its bid given the probabilistic reference prices supplied by the facilitator. We represent these prices, having the form of time-dependent probability distributions, by independent random variables $X_{t}$ with bounded expectation $\mathbb{E}(X_{t})=\bar{x}_{t}<\infty$. 
            Each device agent aims at minimizing its expected cost out of self-interest. For that, we develop a MDP model for optimal bidding which consists of a state space, action space and a set of rewards/costs. We show that the MDP-based bidding algorithm minimizes the expected cost for the device (i.e.\ optimal in expectation) given the available information (i.e.\ probabilistic price reference) and the assumption that a single device is a price taker. For an uninterruptible deferrable device, the action space only consists of two actions \texttt{on}, \texttt{off}. The state consists of a possible realization of the price, and the status ($s_{t}^{a}$) of the device, where $s_{t}^{a}=0$ for a device that has not started yet (i.e.\ \texttt{waiting}), $s_{t}^{a}=\{1,\ldots,D^a-1\}$ for a device that has started (i.e.\   \texttt{running}), and $s_{t}^{a}=D$ for a device that has run for $D$ time-steps (i.e.\ \texttt{finished}). 
            
             If the uninterruptible device $a$ switches from the \texttt{waiting} to the \texttt{running} state at time $t$ with a market clearing price $x_t$, its expected total running cost is a combination of the cost of starting at $t$ with a price of $x_t$, and the sum of the expected costs for the remainder of the device's cycle,
            \begin{equation}
                C^{s,a}_t(x_t) = x_t \cdot P_{0}^{a} \cdot \Delta t
    		                +
            \sum_{i=1}^{D^{a}-1} 
    						\bar{x}_{t+i}			
    						\cdot P_{i}^{a} \cdot \Delta t
		            \label{eq: startC}
            \end{equation}
            The agent aims to minimize its running cost. It does so by, at each time step, submitting a bid function $b^a_t(x)$, defined by a threshold price $\hat{x}_{t}^{a}$. The definition of an optimal bid function is given below.
        	
            \begin{theorem} \label{th:bids}
                For a sequence of independent reference prices $X_t$ with bounded expectation, agent $a$ minimizes its expected running cost by submitting the threshold-based bid function $b_{t}^{a}(x)$, where
		        \begin{align}
			        b_{t}^{a}(x)= & \begin{cases}
				    P^{a}	& x	\leq \hat{x}_{t}^{a}	\\
				    0 		& x	>	 \hat{x}_{t}^{a}	\\
			        \end{cases}, \\
			        P^{a}=& \begin{cases}
			        P_{s_{t}^{a}}^{a}           &   \text{if} s_{t}^{a} <D^{a}\\
			         0                               &   \text{otherwise}\\
			        \end{cases}\\
			        \hat{x}_{t}^{a}=& \begin{cases}
				    -\infty		&	\text{if}\; s_{t}^{a}=D^a	            \\
				    \infty		&	\text{if}\; s_{t}^{a}={1,\ldots,D^a-1}	\\
				    z_{t}^{a}	&	\text{if}\; s_{t}^{a}=0             	\\
			        \end{cases},
			    \label{eq: threshold1}			\\
                    z_{t}^{a}=& \begin{cases}
				    \infty		                    &	t \ge d^{a}-D^{a} 		\\
				    \frac{
					    C^{*a}_{t+1}
					    -
					    \sum_{i=1}^{D^{a}-1}
					    \bar{x}_{t+i} \cdot P_{i}^{a} \cdot \Delta t
				    }
				        {
					P_{0}^{a} \cdot \Delta t
				    } 			                    &	t<d^{a}-D^{a}	\\
			        \end{cases},
			        \label{eq: threshold2}			\\
		        \intertext{and $C^{*a}_{t}$ is the optimal expected cost at $t$, which is recursively defined in reverse order for $t \le d^{a}-D^{a}$ by}
		        	C^{*a}_{d^{a}-D^{a}}= & \sum_{i=0}^{D^a-1}\bar{x}_{d^{a}-D^{a}+i}\cdot P_{i}^a \cdot \Delta t , \label{eq: lastC} \\
    					C^{*a}_{t}= &   				    \mathrm{Pr}(X_{t}>\hat{x}_{t}^{a})\cdot C^{*a}_{t+1}  \nonumber	\\
    				&	+
    						\mathrm{Pr}(X_{t}\leq \hat{x}_{t}^{a}) \cdot \mathbb{E}\left[C^{s,a}_t (X_{t})|X_{t}\leq \hat{x}_{t}^{a}\right] . 
    			    \label{eq: optimalC}			
		        \end{align}
	        \end{theorem} 
	        
	        \begin{proof} 
                In order for $b^a_t(x)$ to be optimal, the optimal action for an agent must be \texttt{on} if the clearing price $x_{t}$ is smaller than or equal to the threshold bid $\hat{x}_{t}^{a}$, and \texttt{off} if it is larger than the threshold bid. First, if $s^a_t =D $ (i.e.\ \texttt{finished}), the only feasible, thus optimal, action is \texttt{off} regardless of the price ($b^a_t(x)= 0 $, i.e.\ $\hat{x}_{t}^{a}=-\infty$). Similarly, if $s^a_t={1,\ldots D^{a}-1} $ (i.e.\ \texttt{running}), and has not completed its task, the only feasible, thus optimal, action is \texttt{on} regardless of the price ($b^a_t(x)= P_{i}^a $, i.e.\ $\hat{x}_{t}^{a}=\infty$), where $i$ is the relevant time period in the device’s program.
	            Finally, a \texttt{waiting} device has different optimal actions based on the following logic.
	            
	            \begin{itemize}
	                \item At time-step $t=d^{a}-D^{a}$, a \texttt{waiting} device $a$ must switch to the \texttt{running} state to meet the deadline, so the optimal action is \texttt{on} irrespective of the clearing price (i.e.\ $\hat{x}_{t}^{a}=\infty$). The expected cost associated with starting immediately is therefore also optimal: $C^{*a}_{d^{a}-D^{a}}=  \mathbb{E}\left[C^{s,a}_{d^{a}-D^{a}} (X_{d^{a}-D^{a}})\right]$, resulting in \eqref{eq: lastC}.
	                
	                \item At time-steps $t < d^{a}-D^{a}$, if $a$ has not started yet, the action \texttt{on} is optimal when the expected cost for switching on is less than the expected cost for waiting and acting optimally at $t+1$, that is, if $C^{s,a}_t(X_t)  < C^{\mathrm{*a}}_{t+1}$. Conversely, if $C^{s,a}_t(X_t) > C^{\mathrm{*a}}_{t+1}$, only \texttt{off} is optimal. Therefore, the threshold $z^a_t$ for $t < d^a-D^a$ in \eqref{eq: threshold2} is derived from the equality
	                \begin{equation} \label{eq:thresholdequality}
	                    C^{s,a}_t(\hat{x}^a_t) = C^{\mathrm{*a}}_{t+1}.
	                \end{equation}
	                When the equality holds, agent $a$ is indifferent between starting and waiting. 
	                
	            \end{itemize}
                Given the existence of optimal threshold bids $\hat{x}^a_t$ and \eqref{eq: lastC}, the optimal expected cost  \eqref{eq: optimalC} for $t < d^{a}-D^{a}$ follows by backwards induction.
            \end{proof}
            
            In the following, we consider how different deadlines impact the bids of otherwise identical agents. Identical devices pose a challenge due to the increased possibility for synchronised and conflicting decisions~\cite{mutually,heterogenity_instability,fast_bidding,ancillary,TCL}. We argue that F-MBC provides a natural way to resolve such conflicts.
	        
	        In the proofs, we shall assume that at any time, the forecast price has a non-zero probability to exceed the largest finite threshold price: $\mathrm{Pr}(X_t > \hat{x}^a_t) > 0,\,\forall a, \forall t: t < d^a-D^a$. 
	        Practically, this means remaining in a \texttt{waiting} state is always an option, unless an agent is forced to start by an upcoming deadline.\footnote{If this condition does not hold for a given pair $\{a,t\}$, agent $a$ concludes that it is always optimal to start at time $t$ (or at an earlier time), i.e.\ for all possible realisations of the random clearing price $X_t$. This effectively adjusts the deadline $d^a \rightarrow \tilde{d}^a = t+D^a$, thus removing the differentiation in threshold bids among affected devices. We note that this is desirable behaviour if the forecaster correctly identified the range of $X_t$, but may cause problems if this range was underestimated, hence including a non-vanishing tail probability in the forecast is recommended.}
	        
	        \begin{definition}
            Agents $\{1,\ldots,n\}$ are \emph{rapid-starting, identical and deadline-ordered} if their power requirement and  service duration are identical and they start consuming immediately ($ \forall a,i: D^a\equiv D, P^a_i \equiv P_i, P^a_0\neq 0$), but their deadlines satisfy $d^1 < d^2 < \ldots < d^n$. They are \emph{weakly} deadline-ordered if their deadlines satisfy $d^1 \le d^2 \le \ldots \le d^n$.
            \end{definition}

            \begin{lemma}
	            \label{lem: diversity}
	            A collection of $n$ rapid-starting, identical, deadline-ordered devices that is in the \texttt{waiting} state at time $t$, operating under the optimal MDP policy, will bid with a strictly decreasing sequence of threshold prices: $\hat{x}^1_t > \hat{x}^2_t > \ldots > \hat{x}^n_t$. A weakly deadline-ordered collection will bid with a non-increasing sequence of threshold prices: $\hat{x}^1_t \ge \hat{x}^2_t \ge \ldots \ge \hat{x}^n_t$.
	        \end{lemma}
	        \begin{proof} 
	            Contained in~\ref{sec:diversityproof}.
	        \end{proof}

	        \begin{theorem}
            \label{th: diversity}
	        A collection of $n$ rapid-starting, identical, deadline-ordered devices, operating under the optimal MDP policy, will start (and complete) in order of their deadlines.
	        \end{theorem}
	        \begin{proof}
	        Prior to the first auction, all agents are in the \texttt{waiting} state. In the auction, agents with a threshold bid exceeding (and sometimes including) the clearing price transition to the \texttt{running} state. \cref{lem: diversity} guarantees that these are agents with the earliest deadlines. This process is repeated for subsequent auctions with devices that have not started yet.
	        \end{proof}
	        
        \subsection{Market Clearing and Tie Breaking}
            The market is cleared via a one-shot double auction for each time-step. We assume that generation truthfully reveals its marginal cost function. The aggregate offer function accounts for flexible generation and inflexible generation in the upcoming time-step in the form of a marginal cost function. Device agents submit their bids only for the upcoming time-step. The aggregate bid function includes inflexible demand and the bids submitted by flexible demand. The market is cleared at time-step $t$ at the price $x_{t}$ at which supply meets demand. Then, the market-clearing price is communicated to device agents which determine their local control actions based on their earlier submitted bids.
            
            Although \cref{lem: diversity} ensures differentiation of bids among devices with different deadlines, equal bids may be submitted, for example if identical devices have identical deadlines. A tie situation occurs when the market clears at the price bid by multiple agents, $x_t = \hat{x}^a_t = \hat{x}^b_t = \ldots$. The aggregate bid/offer functions for such a case are shown in \cref{fig: tie}. A large step in the aggregate bid can cause difficulties in market clearing (i.e.\ bulk switching). To address this issue, we introduce a tie breaking mechanism among such agents.
            
            The tie breaking mechanism determines which of the tied agents can start at the current time-step and which will wait for a later time-step. Each agent submits a random number $\rho^a$ along with its bid. When the auctioneer detects a tie situation, it determines a value $\rho^*$ so that only bids with $\rho^a \le \rho^*$ will be accepted. $\rho^*$ is chosen such that demand most closely approximates the supply at the clearing price $x_t$. 
            
            Due to the discrete nature of the loads, an exact match may not be found. In such a case, the bid of the marginal device $a$ is accepted with probability $\frac{\gamma}{P^{a}}$, where $\gamma$ is the difference between the supply at $x_t$ and the demand without the marginal device. Agents will be charged the market clearing price $x_{t}$ while generation should supply at a slightly higher (lower) set-point, and is paid accordingly. This results in a budget imbalance that vanishes in expectation (i.e.\ averages to zero in the long term). This is illustrated in \cref{fig: tie}.
            \begin{figure}
                \includegraphics[width=\linewidth,keepaspectratio]{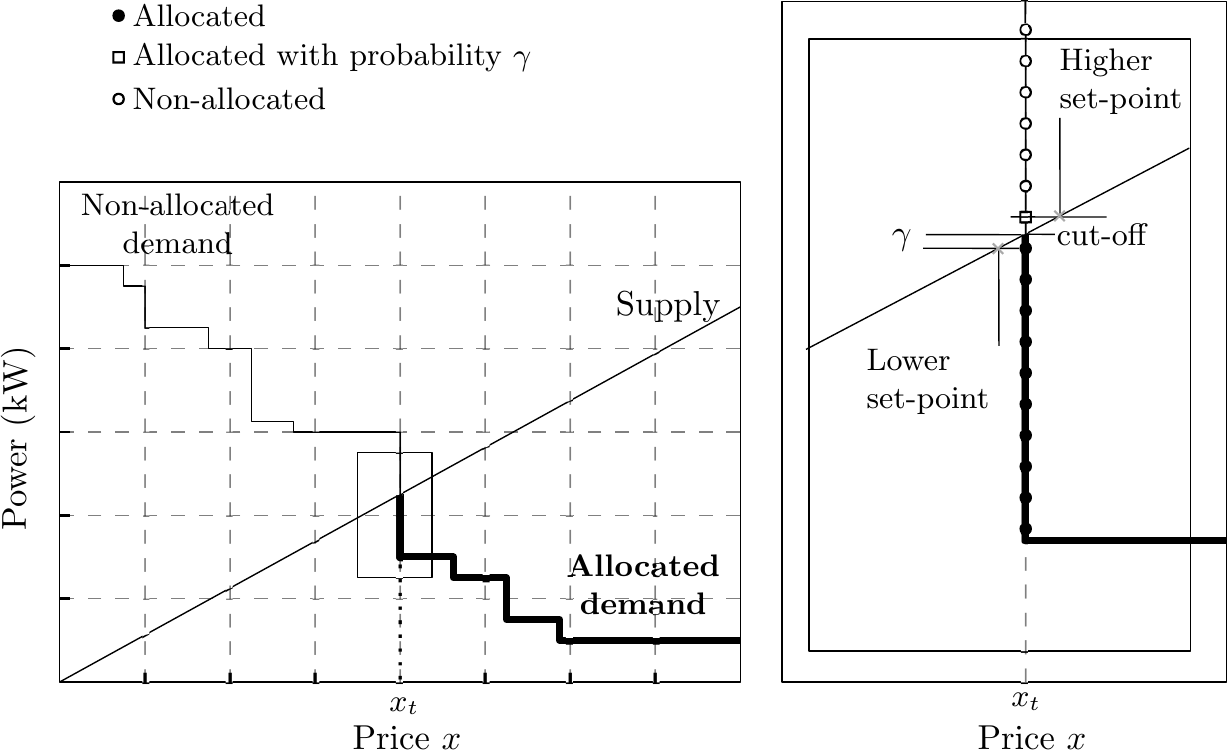}
                \caption{A tie situation; devices that are indifferent between starting or waiting at $x_t$ are allocated randomly.}
                \label{fig: tie}
	        \end{figure}
	        
	        We note that the random tie breaking mechanism does not affect optimality or fairness as it is only used to break ties among agents with bids that are equal to the market clearing price. Agents are indifferent between starting and waiting at their bid price, so those who are not allocated will wait for a later time-step and eventually incur the same expected cost as those which were allocated. Therefore, they have no incentive to game the tie-breaking mechanism. Also, because $\rho^a$ is generated locally, tie breaking can be implemented using a broadcast of $\rho^*$. The alternative, where $\rho^a$ is determined by the auctioneer, would require a targeted message to each device.
	        
            \subsection{Alignment of Optimal Coordination and Self Interest}
            According to the previously stated definition of the optimal coordination problem, our objective is to steer the cluster of flexible devices towards an optimal system-level behaviour (i.e.\ total generation cost minimization). To guarantee a stable optimum, it is necessary that the optimal coordination corresponds to a Nash equilibrium. This guarantees that it is in the best interest of the device agents not to deviate from such behaviour. Therefore, we show that the global cost minimizing solution indeed corresponds to a Nash equilibrium. To analyse the potential for F-MBC to achieve optimal system-level behaviour, we first consider the schedule achieved by a clairvoyant optimizer with complete information. We give conditions under which this schedule corresponds to the outcome of a Nash equilibrium, i.e.\ agents cannot benefit by deviating from the starting time-step allocated by the central optimizer. These results  indicate that the central F-MBC facilitator should aim to estimate the prices that correspond to such a system optimal allocation, so that devices are incentivised to realise the reference prices. 
            
             We consider a cost-optimal allocation of flexible devices, characterized by an aggregate load profile $P^*_t$ and a starting time $t^a$ for each flexible device $a$, summarized as $\mathcal{S} = \left(\{P^*_{t}\}_{t=1:\text{T}} ,\{t^a\}_{a=1:A}\right)$. Without loss of generality, in the following we take the perspective of an arbitrary deferrable device agent $a$ that has a duration $D$ and uninterruptible consumption pattern $\{P^a_0, \ldots ,P^a_{D-1}\}$, which is scheduled to start at $t=t^a$ under the cost-optimal allocation. No assumptions are made about the properties of other flexible loads. Let $P^{\lnot a}_t$ be the cost-optimal load pattern $P^*_t$ minus the consumption of device $a$ starting at $t^a$, and $m_t(P)$ be the monotone increasing function in $P$ which represents the marginal cost of a unit of generation at generation level $P$ and time $t$. The cost to the system of running device $a$ at time $t$ is 
            \begin{equation} \label{eq:Kdefinition}
                K^a_t = \Delta t \sum_{i=0}^{D-1} \int_{P^{\lnot a}_{t+i} }^{P^{\lnot a}_{t+i}  + P^a_i} m_{t + i}(P) \mathrm{d}P.
            \end{equation}
            The fact that the starting time $t^a$ is optimal with respect to overall system cost, implies that
            \begin{equation} \label{eq:Kcondition}
                K_{t^a}^a \le K_t^a, \qquad \forall t \in \mathcal{T}.
            \end{equation}
	        Switching from the system perspective to that of an individual, we assume that the $a$ pays a price equal to the marginal cost of energy. The total price paid by agent $a$ starting at $t$ is
	        \begin{equation} \label{eq:Pidefinition}
                \Pi^a_{t} = \Delta t \sum_{i=0}^{D-1} m_{t + i}(P^{\lnot a}_{t+i} + P^a_i) P^a_i.
            \end{equation}
	        The allocation $\mathcal{S}$ is a Nash equilibrium if for each agent $a$,
            \begin{equation} \label{eq:NEcondition}
                \Pi_{t^a}^a \le \Pi_t^a, \qquad \forall t \in \mathcal{T}.
            \end{equation}
            In the following, we identify conditions where global cost-optimality \eqref{eq:Kcondition} implies the Nash equilibrium condition \eqref{eq:NEcondition}. We first consider a (restrictive) special case in which the implication holds exactly; we then consider a weaker set of conditions that results in \cref{NE2} and \cref{col:largeNE} with much broader applicability.
            
            \begin{theorem}
                If $m_t(P)$ is an affine function with constant slope $\mathrm{d}m_t(P)/\mathrm{d}P=c, \forall t$, then $\mathcal{S}$ is a Nash equilibrium.
                \label{affine}
            \end{theorem}
            \begin{proof}
                Evaluating the integral in \eqref{eq:Kdefinition} using the affine structure of $m_t(P)$ yields
                \begin{equation}
                    K_t^a = \Pi_t^a - \frac{c \Delta t}{2} \sum_{i=0}^{D-1}(P^a_i)^2.
                \end{equation}
                Because the last term does not depend on $t$ or $P^*$, \eqref{eq:Kcondition} implies \eqref{eq:NEcondition} and $\mathcal{S}$ is a NE.
            \end{proof}
            Note that $m_t(P)$ does not need to be strictly affine with slope $c$ for all $P$, but only for those marginal power levels that are accessible by flexible devices. This is the case in the example in \cref{sec: result}.

            \begin{definition}
                The allocation $\mathcal{S}$ is a $\delta$-relaxed Nash equilibrium if the condition \eqref{eq:NEcondition} is replaced by the weaker condition 
                            \begin{equation}  \label{eq:deltaNE}
                    \Pi_{t^a}^a \le (1+\delta) \Pi^a_{t} , \qquad \forall t \in \mathcal{T}.
                \end{equation}
            \end{definition}
            
            The $\delta$-relaxed Nash equilibrium is effectively a Nash equilibrium for devices that are insensitive to relative price differentials of size $\delta$. Clearly, it converges to a regular Nash equilibrium in the limit $\delta \downarrow 0$. We note that this is closely related to the concept of an $\varepsilon$-equilibrium~\cite{AGT}.

            \begin{theorem} \label{NE2}
                If there exists an $\varepsilon < 1$ so that,
                \begin{multline} \label{eq:smalldifference}
                    m_t(P^{\lnot a}_{t} + P^a_i) - m_t(P^{\lnot a}_{t}) \le \varepsilon m_t(P^{\lnot a}_{t} + P^a_i), \\ \forall t \in \mathcal{T}, \forall i \in \{0,\ldots, D-1 \},
                \end{multline}
                then $\mathcal{S}$ is a $\delta$-relaxed Nash equilibrium with $\delta = \varepsilon/(1-\varepsilon)$
            \end{theorem}
            \begin{proof}
                From definitions \eqref{eq:Kdefinition}, \eqref{eq:Pidefinition}, \eqref{eq:smalldifference} and the fact that $m_t(P)$ is non-decreasing, it follows that 
                \begin{equation}
                    (1-\varepsilon) \Pi^a_{t} \le K^a_{t} \le \Pi^a_{t}, \qquad \forall t \in \mathcal{T}.
                \end{equation}
                By chaining the first inequality (for $t=t^a$) with \eqref{eq:Kcondition} and the second inequality, we obtain
                \begin{equation}  \label{eq:lemma-result}
                    (1-\varepsilon)\Pi_{t^a}^a \le \Pi_t^a , \qquad \forall t \in \mathcal{T}.
                \end{equation}
                Substitution of $\delta=\varepsilon/(1-\varepsilon)$ and comparison with \eqref{eq:deltaNE} completes the proof.
            \end{proof}
            
            \begin{corollary} \label{col:largeNE}
            In the limit where agents are price takers (individually), $\mathcal{S}$ is a Nash equilibrium. 
            \end{corollary}
            \begin{proof}
                When agents are price takers their influence on $m_t$ is negligibly small; this implies that  Theorem~\ref{NE2} applies with the limit $\varepsilon \downarrow 0$. Therefore, $\delta \downarrow 0$ in \eqref{eq:deltaNE} and the stronger condition \eqref{eq:NEcondition} holds.
            \end{proof}
            This result effectively extends the Nash equilibrium to to all sufficiently large systems with continuous marginal cost functions. Note that the notion of individual device agents being price takers does not preclude devices from \emph{collectively} influencing prices significantly. 
            
        \subsection{Approximate consistency of solutions} \label{sec:consistency}

        In this section we quantify the consistency of the proposed F-MBC approach. Ideally, if the facilitator is able to supply the agents with reference prices that are realisable and near-optimal, the agents should respond with bids that result in start times consistent with that profile. If this is the case, the (near-)Nash Equilibrium that is encoded in the reference prices would become \emph{self-fulfilling}. In order to quantify this property, we investigate deviations from the optimal coordination solution in the limit where the reference prices correspond to such a solution. We do so for the special case of a collection of rapid-starting, identical devices, and a single time step $t$. In the following, superscripts $a$ are dropped for identical quantities (e.g. durations $D$). Near-optimality of price forecasts is represented by price forecasts $X_{t'} = m_{t'}(P^*_{t'}) + \Delta_{t'}$ for $t'>t$, where $P^*_{t:T}$ represents a feasible cost-optimal consumption schedule and the magnitude of $\Delta_{t:T}$ is strictly bounded from below.  
        
        \begin{lemma} \label{lem:pricediff}
        Consider a collection of rapid starting devices operating under the F-MBC framework, and a cost-optimal allocation $\mathcal{S}$, characterised by a starting time $t^a \ge t$ for each device $a$, and an aggregate load $P^*_{t:T}$ (including inflexible load). 
        
        If devices receive near-optimal reference prices $X_{t'} = m_{t'}(P^*_{t'}) + \Delta_{t'}$, where $E[\Delta_{t'}]=0$ and $\mathrm{Pr}(\Delta_{t'} > - \eta)=1$ for $\eta > 0$ for all $t'\in \{t,\ldots, T  \}$, then the difference between the clearing price $x_t$ and the reference price $x_t^*=m_t(P^*_t)$ is bounded by 
        \begin{multline}
            - \frac{\sum_{i=0}^{D-1} P_i \left[\Delta m_{t,i} + \eta \right]}{P_0} \le x_t - x_t^*  \le  \\ 
             \max_{t' \in \{t+1,\ldots,  T\}} \frac{\sum_{i=0}^{D-1} P_i \Delta m_{t',i}}{P_0} \label{eq:pricebounds}
        \end{multline} 
        with
        \begin{equation} \label{eq:deltamdef}
            \Delta m_{t,i} = m_t(P_t^*) - m_t(P_t^* - P_i)
        \end{equation}

        \end{lemma}
        \begin{proof}
        Contained in~\ref{sec:pricediffproof}.
        \end{proof}
        
        Let $n^*_t$ be the number of devices starting at $t$ under the optimal allocation $\mathcal{S}$, and $n_t$ the number of devices starting using the F-MBC dispatch method. Under the additional assumption of smoothness of the marginal cost of generation, it is possible to derive bounds for the difference $n_t - n_t^*$, as follows. 
        
        \begin{theorem}
        Assume that \cref{lem:pricediff} holds, and that $m_t(P)$ can be approximated around $P^*_t$ by
        \begin{equation}
            m_t(P) = m_t(P^*_t) + c_t \left[ P - P_t^* \right].
        \end{equation}
        Then
                \begin{multline}
            - \left\lceil \frac{\sum_{i=0}^{D-1}  \left[c_{t+i} P_i^2 + \eta P_i \right]}{c_t P_0^2} \right\rceil \le n_t - n_t^*  \le  \\ 
             \max_{t' \in \{t+1,\ldots,  T\}} \left\lceil \frac{\sum_{i=0}^{D-1} c_{t'+i} P_i^2 }{c_{t} P_0^2} \right\rceil. \label{eq:numberbounds}
        \end{multline} 
        \end{theorem}
        \begin{proof}
        The optimal clearing price $x_t^*$ is associated with the desired number of starting devices $n_t^*$. Linearity of the marginal cost function and rounding up/down to the nearest integer results in 
        \begin{equation}
            \left\lfloor \frac{x_t - x^*_t}{c_t P_0}  \right\rfloor \le n_t - n_t^*  \le \left\lceil \frac{x_t - x^*_t}{c_t P_0}  \right\rceil
        \end{equation}
        Combining with \eqref{eq:pricebounds} and making use of 
        \begin{equation}
            \Delta m_{t,i} = c_t P_i
        \end{equation}
        results in \eqref{eq:numberbounds}.
        \end{proof}
        
        \begin{corollary}
        In the special case where $m_t(P)$ is an affine function with constant slope $c_t=c$, devices consume a constant amount of power ($P_i=P$) and in the limit of vanishing uncertainty ($\eta \downarrow 0$), we have
        \begin{equation}
        -D-1 \le n_t-n_t^* \le D.
        \end{equation}
        \end{corollary}

        These results show that the F-MBC dispatch converges to the optimal dispatch within hard limits. These limits do not depend on the total number of devices, so the relative performance increases with the number of devices. 
        
        Moreover, the analysis above considers only a single time step $t$. If the number of devices $n_t$ starting at $t$ exceeds $n_t^*$, this results in higher prices for subsequent time steps, thus reducing the number of devices that start at $t+1$ until $t+D-1$. Conversely, if the number of device starts is lower than scheduled, this will incentivise additional starts in subsequent time steps. Although not quantified here, this self-regulating effect further reinforces the convergence to the reference solution.

    \section{Experimental analysis} \label{sec: result}
        Using simulations, we illustrate two features of the proposed F-MBC approach. First, we show that F-MBC performance is near-optimal over multiple time-steps when price uncertainty is negligible (consistency). Second, we analyze the robustness of the solution to varying amounts of uncertainty in price forecasts in order to qualify the need for accurate estimation of reference prices.
        
        \begin{figure}
            \centering
            \includegraphics[width=\linewidth,keepaspectratio]{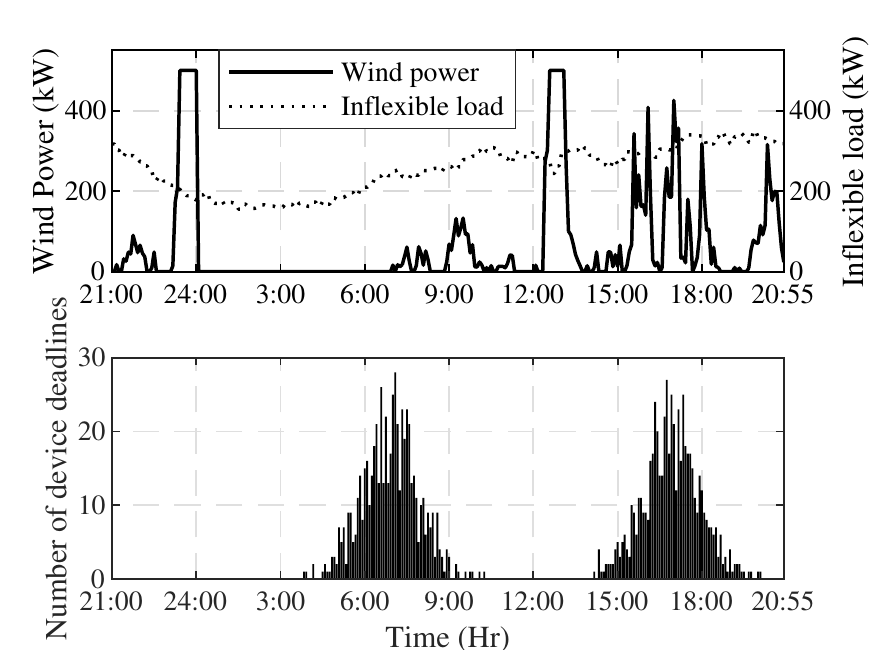}
            \caption{Simulation input data.
                        Top: Wind power generation and inflexible load profile.
                        Bottom: Distribution of device deadlines across the 5 minute time intervals.}
            \label{fig: input}
        \end{figure}
            
	    \subsection{Case Study Description}

		    For this case study, we consider a system with identical deferrable loads. This represents a particularly challenging scenario, due to a high probability of ties occurring and a lumpiness of loads that does not permit full `valley filling' of the solution. A full day (24 hours, starting at 21:00) was simulated with market clearing at \SI{5}{\minute} time steps. A fixed horizon at 20:55 the next day was used for forecasting and bid formation. The system included \SI{1200} deferrable loads, with a duration of \SI{1}{\hour} and fixed consumption of \SI{2}{\kilo \watt} each. Deadlines were distributed in two clusters of \SI{600} devices, normally distributed with a standard deviation of 1 hour around 7:00 in the early morning and 17:00 in the early evening, and rounded to the nearest \SI{5}{\minute} time-step. Inflexible demand was modelled using load data from~\cite{load_data} aggregated and scaled to a peak of \SI{350}{\kilo \watt}. Wind generation with a peak of \SI{500}{\kilo \watt} was generated using~\cite{wind_data} and a simple wind turbine model that approximates the performance of a \SI{100}{\kilo \watt} wind turbine~\cite{wind}, scaled to \SI{500}{\kilo \watt}. We assume that wind power generation is free and curtailable. The simulation input data can be seen in \cref{fig: input}. Simulations were performed in Matlab.

		    Flexible generation was represented by a time-independent, linearly increasing marginal cost function
	        \begin{equation} \label{eq: marginal}
                m(P^{\mathrm{g}})
                =
                \frac{
                    P^\mathrm{g}
                }{
                    k
                } 
                ,\;
                P^\mathrm{g} \geq 0,
            \end{equation}
            where $P^\mathrm{g}$ is the power generated by the flexible generator fleet and $k = \SI{500}{\kilo \watt^2 \minute}$, with arbitrary units for currency. With this choice, the total cost of generation (wind and flexible generation), has an affine marginal cost, provided that $P^\mathrm{g}>0$. Therefore, it follows from \cref{affine} that the device schedule from a clairvoyant optimizer with complete information corresponds to the outcome of a Nash equilibrium.
            
        \subsection{Simulating the Facilitator: Clairvoyance and complete control}
            To establish the potential of F-MBC as a coordination mechanism via simulations, we first identify the theoretical optimal coordination that can be obtained only by a clairvoyant optimizer with complete control. Accordingly, we obtain optimal reference prices that reflect an optimal allocation of demand using perfect foresight. Due to our selection of identical, fixed consumption devices, this can be done by solving the mixed-integer quadratic program (MIQP) that finds the optimal number of devices to start at each time-step $\sigma_{t}$, and optimal flexible power generation for each time-step $P^\mathrm{g}_{t}$ such that the total generation cost over multiple time-steps is minimized:
		    
            \begin{align}
    		    &	\underset{P^{\mathrm{g}}_{t},\sigma_{t},o_{t}}{\text{minimize}}  \quad   \sum_{t \in T}  \frac{1}{2} \cdot \frac{(P^{\mathrm{g}}_{t})^{2}}{k} \cdot \Delta t,
    			\label{objective} \\
                \intertext{subject to, $\forall t \in T$,}
    			& P^{\mathrm{g}}_{t}   \geq 0, \label{const.1} \\
    			& P^{\mathrm{g}}_{t}+P^{\mathrm{r}}_{t}   \geq o_{t} \cdot P^{a} + P^{\mathrm{l}}_{t}, \label{const.2} \\
    			& \sum_{i=1}^{t} \sigma_{i}   \geq \phi_{d} (t+D) ,\; t \leq \text{T}-D,   \label{const.3}\\
    		    &\sum_{i=1}^{t} \sigma_{i}   =\phi_{d} (\text{T}) ,\; \text{T}-D+1\leq t \leq \text{T} \label{const.4} \\
    			& \sigma_{t}  \geq 0 ,\label{const.5} \\
    			& o_{t}= \begin{cases}
    			\sum_{j=1}^{t} \sigma_{j} 			& 	t \leq D	\\
    			\sigma_{t}+(o_{t-1}-\sigma_{t-D}) 	&	t > D		\\
    			\end{cases}
    			\label{const.6}
	        \end{align}
	        where at time-step $t$, $\phi_{d}(t)$ is the number of devices with deadlines before or at $t$, $P^{\mathrm{l}}_{t}$ is power consumption by inflexible load, $P^{\mathrm{r}}_{t}$ is the power from renewable sources, $o_{t}$ is the number of devices running at $t$. The objective function in \eqref{objective} is the integral of the marginal cost \eqref{eq: marginal}.
    	    Generator limits and supply/demand matching constraints are represented by \eqref{const.1} and \eqref{const.2}, respectively. It is assumed that renewable generation is curtailed when a generation surplus occurs. The number of device start-ups to any time-step $t$ must be at least equal to the number of devices which have a deadline before or at $t+D$, and for the last time periods it should be exactly the total number of devices with a deadline before $T$. This is represented by \eqref{const.3}-\eqref{const.4}. Device uninterruptibility is ensured by \eqref{const.5}-\eqref{const.6}. Combined, \eqref{const.3}-\eqref{const.6} guarantee that devices will not miss their respective deadlines. By solving the MIQP, a cost-optimal system load profile is obtained, which corresponds to a set of reference prices $x^{*}_{t}\;\forall t \in T$. While the reference solution here does not account for specific allocation for each agent, one realization of the reference schedule can be achieved by giving priority to devices according to their proximity to their respective deadlines, with ties being broken randomly, and assuming that devices do not switch off until their cycle (duration) is complete. The optimization was repeated after each market clearing to account for deviations from the previous reference solution, to effectively generate an ``up-to-date'' forecast at each time step.

    	\subsection{Simulating the Facilitator: bounded information and Uncertainty}
    	    As previously established, probabilistic reference prices are required. In reality, the facilitator would provide probabilistic reference prices that depend on actual forecasts and information used in generating the reference. Instead, for simulation purposes, probabilistic price forecasts were generated by adding noise to the deterministic reference prices $x^{*}_{t}$ as follows. It was assumed that uncertainties are exogenous and independent for each time-step, and forecast prices at each time-step are  log-normally distributed, with a standard deviation that increases with time. The standard deviation of the price $X_{t}$ as forecast at $t' \le t$ is parametrised by the day-ahead uncertainty $\nu^{24h}$ as
            \begin{equation}
                SD_t = x^{*}_{t} \cdot \nu^{24h} \cdot \frac{(t - t')}{24 h}.
            \end{equation}
        
            Moreover, forecasting errors were simulated by adjusting the mean of the log-normal forecasts: the expected prices $\bar{x}_{t}$ were sampled from the log-normal distribution with mean $x^{*}_{t}$ and standard deviation $SD_t$. The values $\bar{x}_{t},SD_t\, \forall t \in \mathcal{T}$ were communicated to agents to be used for bid formulation.
            
        \subsection{Results} \label{sec: sensitivity}
            \cref{fig: output} shows simulation results obtained with negligible uncertainty ($\nu^{24h} = 10^{-5}$), demonstrating that F-MBC achieves near-optimal performance over multiple time-steps. It can be seen from top panel that the device schedule obtained by F-MBC closely resembles the schedule obtained from the MIQP reference solution. The difference in total generation cost in this case is \SI{0.08}{\percent} compared to the reference solution.
            
             \begin{figure}[t]
        	    \centering
        	    \includegraphics[width=\linewidth,keepaspectratio]{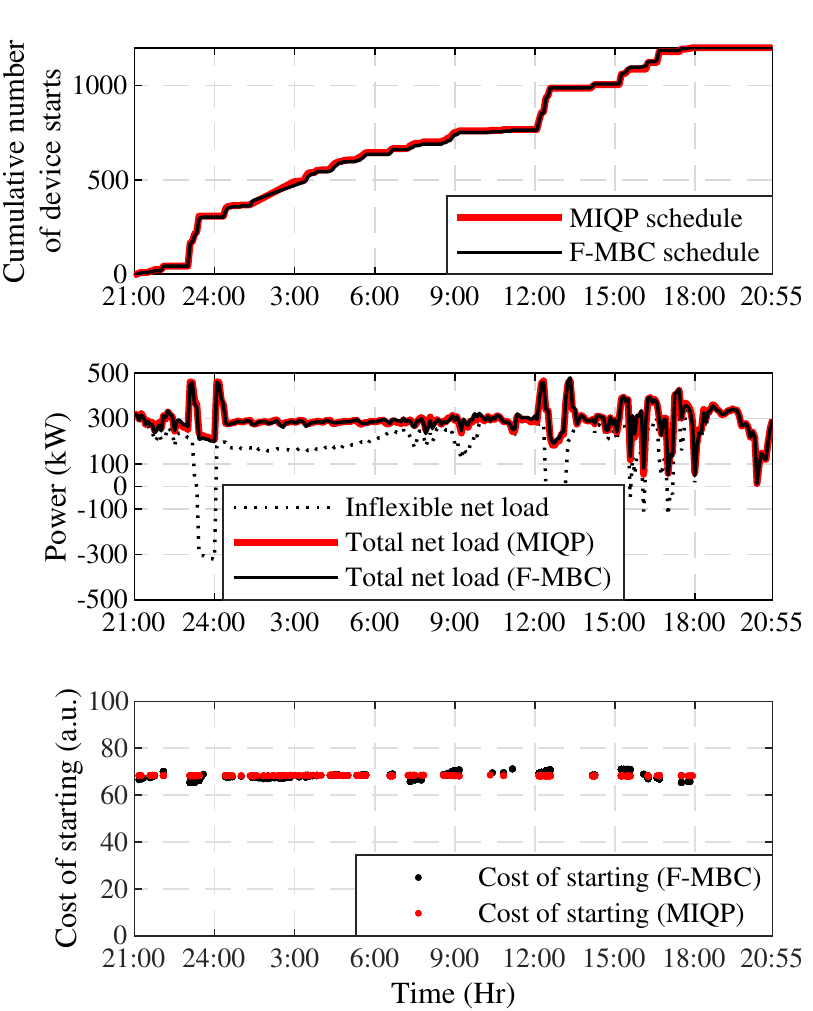}
        	    \caption{Simulation output. Top: Cumulative device starts obtained by F-MBC compared to the MIQP reference solution. Middle: Load minus inflexible generation (i.e.\ power supplied by flexible generation) for three scenarios: without flexible loads, MIQP reference solution and F-MBC solution. Bottom: The realized total cost paid by devices plotted as a function of their starting times.}
        	    \label{fig: output}
        	\end{figure}
        	
            Moreover, the centre panel shows the approximate `valley filling' behaviour of the solution, especially compared to the system without flexible demand (dotted line). We note that  perfect flattening of flexible power generation is not feasible due to the extended run time (1 hour, i.e.\ 12 time steps) of loads. In the bottom panel, costs incurred by devices are plotted against their starting times. The actual costs obtained using F-MBC are very close to the nearly identical costs obtained using MIQP. 

             We compare the performance of F-MBC to three alternative coordination techniques in \cref{fig: methods}, depicting the same information as the lower two panels of \cref{fig: output}. Lack of coordination is represented by the ``latest start'' approach where devices start at the latest time-step possible without missing their respective deadlines. The ``Naive MBC'' approach implements naive agents that submit a bid between the minimum expected price $x^a_{t,min}$ and maximum expected price $x^a_{t,max}$ that occur before their latest start time $d^a-D-1$. The bid placed is $\hat{x}^a_t = x^a_{t,min} + t ( x^a_{t,max}-x^a_{t,min})/(d^a-D^a-1)$, and generally increases as devices approach their deadlines. Moreover, we demonstrate the importance of utilizing probabilistic forecasts by implementing a ``point forecast MBC'' approach, where each agent only receives a time-series of \emph{expected} prices and places an optimal bid using backward induction. This approach performs sub-optimally and yields a total cost error of \SI{9.6}{\percent}. 
     
            	\begin{figure}[th]
        		\centering
        		\includegraphics[width=\linewidth,keepaspectratio]{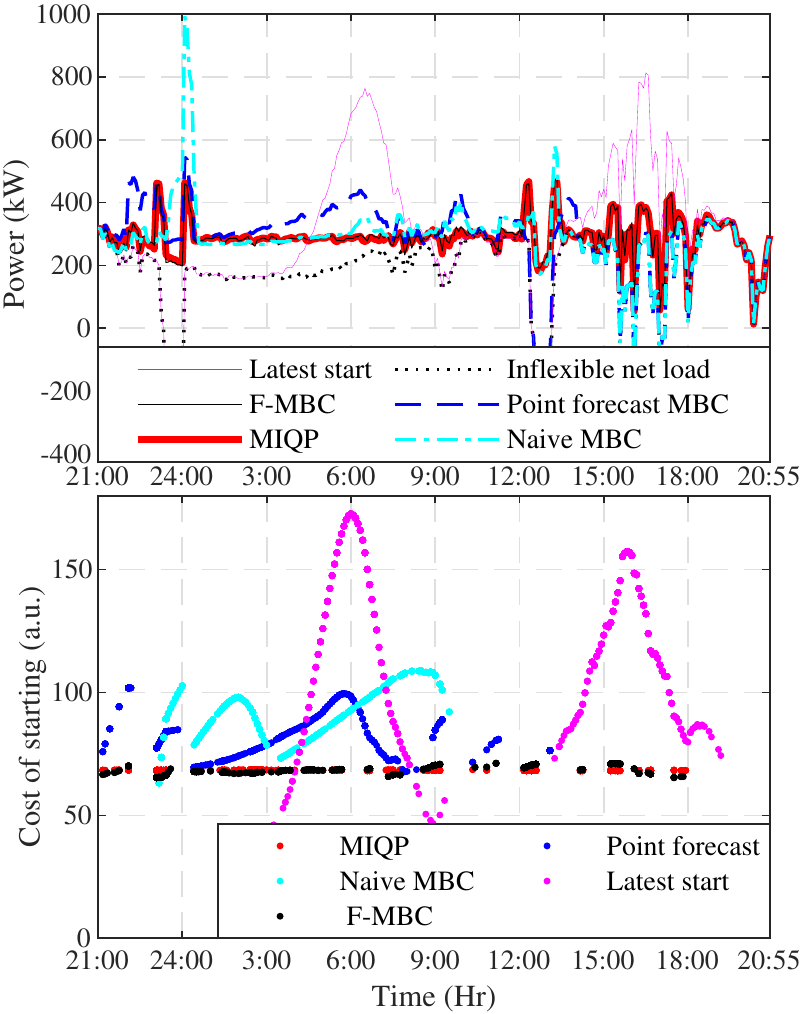}
        		\caption{Comparison of different coordination techniques. Top: Total net load. Bottom: starting cost against starting time per device}
        		\label{fig: methods}
        	\end{figure}

            To evaluate the effect of forecast uncertainty on the performance of the F-MBC approach, we vary $\nu^{24h}$ from $10^{-5}$ to $1$ (i.e.\ 100\%). \cref{fig: sensitivity} shows the results of \SI{20} independent simulation runs for each value of $\nu^{24h}$. The top panel shows the distribution of realised cost of flexible generation, compared with the reference solution. It demonstrates near-optimal performance even for significant uncertainties in forecast prices.  
            
                    	\begin{figure}[th]
        		\centering
        		\includegraphics[width=\linewidth,keepaspectratio]{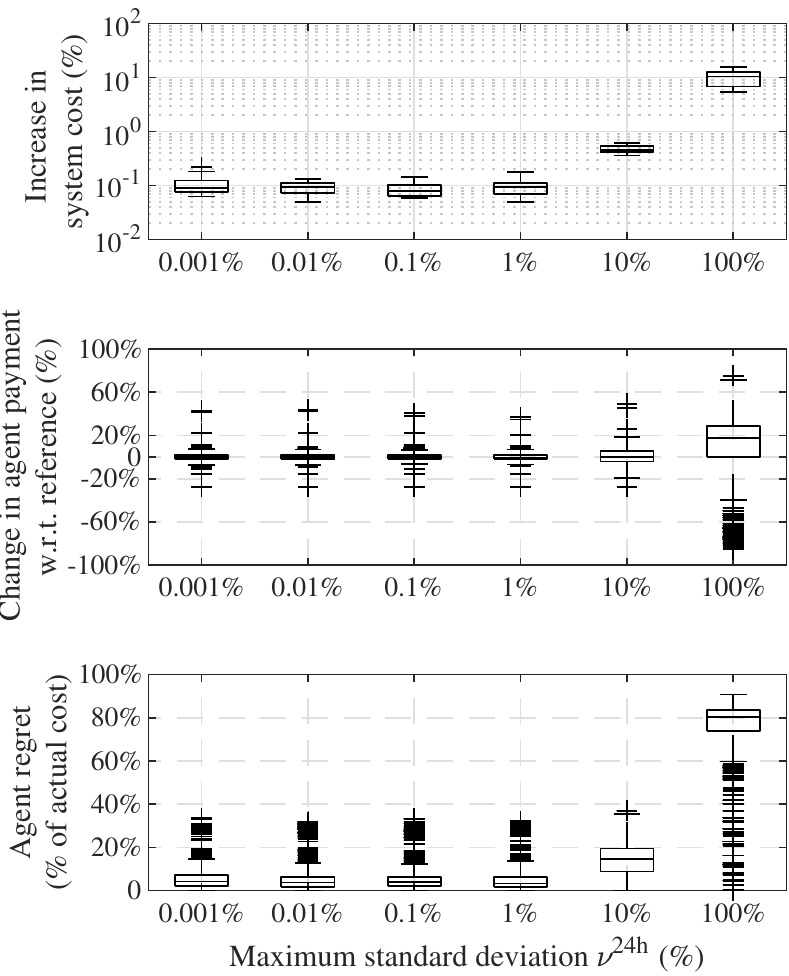}
        		\caption{Sensitivity to uncertainty. Top: Increase in system cost with respect to the optimal solution. Middle: Distribution of  change in agent payments compared to optimal solution. Bottom: Distribution of regret of device agents.}
        		\label{fig: sensitivity}
        	\end{figure}
            The middle panel compares the individual payments made by device agents ($1200 \times 20$ for each value of $\nu^{24h}$) against the payments under the reference schedule. For forecast uncertainties up to 10\%, these are approximately zero-mean, so that devices are on average as well off using the F-MBC coordination scheme as under the Nash equilibrium. 
            
            Finally, the bottom panel depicts the distribution of regret that device agents have as a result of F-MBC (i.e.\ the difference between the actual price paid and the lowest possible price in retrospect). The positive values indicate small deviations from a Nash equilibrium. However, the computed regret can only be used to generate cost savings individually: collectively, devices would quickly equalise price savings, as is evidenced by the small system cost deviations in the top panel. 

        Collectively, these results demonstrate that when supplied with nearly optimal reference prices, F-MBC is able to approximate the optimal schedule, but small differences remain due to the `lumpiness' of load, in line with the approximate consistency results in \cref{sec:consistency}. However, these differences are small when averaged over many runs, and are expected to reduce further as the system size increases.

    \section{Discussion \& Conclusion} \label{sec: conclusion}
        In this paper we considered a setting of uninterruptible deferrable devices with deadlines set by their respective owners. By relying on decentralized decision making and centralized forecasting and market clearing, the F-MBC approach provides a simple and scalable means of DER coordination. In terms of communication infrastructure, the proposed mechanism can be implemented using only gathering of bids and broadcasts of prices and tie-breaking cut-off values from the auctioneer to all devices, significantly reducing implementation complexity. Moreover, since the information in these broadcasts only concerns public information, F-MBC preserves end user privacy and autonomy. The bidding algorithm was shown to automatically resolve mutually-conflicting decisions between devices with different deadlines and a tie breaking procedure was proposed to resolve conflicts between indifferent devices. It was shown by simulation that near-optimal performance can be attained by a clairvoyant facilitator, establishing the consistency of the approach. Moreover, an analysis of the sensitivity to price forecast uncertainty demonstrates the robustness of the approach. It was able to achieve good system-level and device-level performance across an extended horizon, making use of simple agent logic and single-period market clearing.
        
        Prices obtained using F-MBC are determined in real-time, thus exposing users to price uncertainty. The results suggest that the resulting cost fluctuations even out in the long run, so that users are not worse off - especially in comparison with less-optimal schemes. If such exposure is nevertheless undesirable, an alternative is to use F-MBC with a virtual currency, only for coordination and control. A different payment scheme (e.g.\ fixed subscription, average price, etc.\ ) can be operated in parallel.
	   
	    This paper has introduced the F-MBC concept and established its desirable properties in a limited set of applications, thus laying the groundwork for various generalizations. As a proof of concept, we use uninterruptible deferrable loads. However, relevant extensions for future work are the inclusion of heterogeneous sets of deferrable loads, interruptible loads and continuously controllable loads. For example, the charging of electric vehicles can be approximated as one of uninterruptible deferrable loads, so that the results derived in this paper directly apply. However, more elaborate charging models will require extensions to the bidding and clearing algorithms, and are the subject of future work. 
	    
	    In addition, machine learning approaches could be used to generate the forecasts, instead of the stylized approach used here, and performance under the influence of external noise (e.g.\ uncertain wind power output) would be relevant to investigate to better understand the behaviour of F-MBC in practice.

\section*{Acknowledgments}
The authors thank the anonymous reviewers for helpful questions and suggestions that led us to improve this paper.

\begin{appendix}
    \section{Proof of Lemma \ref{lem: diversity}} \label{sec:diversityproof}
        \begin{lemma} For two rapid-starting, deadline-ordered agents in the \texttt{waiting} state, at time-step $t=d^{1}-D$, $C^{*1}_{d^{1}-D}>C^{*2}_{d^{1}-D}$.
	            \label{le: lastC}
            \end{lemma}
            \begin{proof}
                Agent \SI{1} must run at $t=d^{1}-D$, so restating \eqref{eq: lastC}:
		        \begin{align}
        		     & C^{*1}_{d^{1}-D}
        		     =   \nonumber \\
        		     & \left[\mathbb{E}[X_{d^{1}-D}] \cdot P_{0}^{1} \cdot \Delta t
        		     +
        		     \sum_{i=1}^{D^1-1}
        		     \bar{x}_{d^{1}-D+i} \cdot P_{i}^{1} \cdot \Delta t \right].
		        \end{align}
		        Using the definition $p=\mathrm{Pr}(X_{t}>\hat{x}_{t}^{2})$, the optimal expected cost \eqref{eq: optimalC} for agent \SI{2} is rewritten as
		        \begin{align}                
					 C^{*2}_{d^{1}-D}=&	p \left[
    				C^{*2}_{d^{1}-D+1}
    				-
    				C^{s,2}_{d^{1}-D}(\hat{x}^2_{d^{1}-D})
    				\right]	+	p \hat{x}_{d^{1}-D}^{2} \cdot P_{0}^{1} \cdot \Delta t																\nonumber	\\
    				& + (1-p) \mathbb{E}[X_{d^{1}-D}|X_{d^{1}-D}\leq \hat{x}_{d^{1}-D}^{2}]\cdot P_{0}^{1} \cdot \Delta t \nonumber	 \\
    				& + \sum_{i=1}^{D^{1}-1} 
    				\bar{x}_{d^{1}-D+i} \cdot P_{i}^{1} \cdot \Delta t.	
    				\label{eq: long}
		        \end{align}
		        The first term vanishes as a result of the threshold price definition, so the inequality $C^{*1}_{d^{1}-D} > C^{*2}_{d^{1}-D}$ to be proven can be simplified to
		        \begin{align}
        		    & \mathbb{E}[X_{d^{1}-D}] > \nonumber \\
        		    & p \hat{x}_{d^{1}-D}^{2}+
                    (1-p)\mathbb{E}[X_{d^{1}-D}|X_{d^{1}-D}\leq \hat{x}_{d^{1}-D}^{2}].
                    \label{eq: show}
                \end{align}
                Considering $\mathbb{E}[X_{d^{1}-D}]$ as the weighted sum of two conditional expectations (above and below $\hat{x}^2_{d^1-D}$), this can simplified to
                \begin{equation}
		           p \left( \mathbb{E}[X_{d^{1}-D}|X_{d^{1}-D}
                    >
                    \hat{x}_{d^{1}-D}^{2}] - \hat{x}_{d^{1}-D}^{2} \right) > 0.
		        \end{equation}
		        This is positive under the assumption that $p>0$ (prices can exceed $\hat{x}^2_t$). Therefore \eqref{eq: show} holds and  $C^{\mathrm{*1}}_{d^{1}-D}>C^{\mathrm{*2}}_{d^{1}-D}$.
            \end{proof}

            \begin{lemma} \label{le: Cs}
		        For any two rapid-starting, deadline-ordered agents in the \texttt{waiting} state, and $t < d^{1}-D$, if $C^{*1}_{t+1}>C^{*2}_{t+1}$, then $C^{*1}_{t}>C^{*2}_{t}$.
	        \end{lemma}
            \begin{proof}
    	        Analogous to \eqref{eq: optimalC}, we define $C_{t}^{b,a}(\hat{x})$ the expected cost incurred by agent $a$ at time $t$, when it submits a first bid $\hat{x}$ and subsequent optimal bids:
    	        \begin{multline}
                        C_{t}^{b,a}(\hat{x})= \mathrm{Pr}(X_{t}>\hat{x})\cdot C^{*a}_{t+1} 	\\
    					+
    						\mathrm{Pr}(X_{t}\leq \hat{x}) \cdot \mathbb{E}\left[C^{s,a}_t (X_{t})|X_{t}\leq \hat{x}\right] 
	            \end{multline}
                At any time-step $t$, the expected cost incurred by agent \SI{2} for bidding with a price $\hat{x}_{t}^{1}$ is by definition greater than or equal its optimal expected cost,
	            \begin{equation}
                    C_{t}^{b,2}(\hat{x}_{t}^{1}) 
                     \geq
                     C_{t}^{b,2}(\hat{x}_{t}^{2})=C^{*2}_{t}.
                     \label{eq: key1}
                \end{equation}
                The condition $C^{*1}_{t+1}>C^{*2}_{t+1}$, combined with \eqref{eq: key1}, and the assumption $\mathrm{Pr}(X_t > \hat{x}^2_t) > 0$, implies 
                \begin{equation}
                    C^{*1}_{t}
                     =
                     C_{t}^{b,1}(\hat{x}_{t}^{1})
                     >
                    C_{t}^{b, 2}(\hat{x}_{t}^{1}).
                    \label{eq: key2}
                \end{equation}
                Combining \eqref{eq: key1}, \eqref{eq: key2} yields the desired result $
                     C^{*1}_{t}
                     >
                     C^{*2}_{t}$.
            \end{proof}
            
            \emph{Proof of \cref{lem: diversity}}:
            	 Consider two identical, deadline-ordered agents in the \texttt{waiting} state. The ordering of threshold prices $\hat{x}_{t}^{1}>\hat{x}_{t}^{2}$ follows from \eqref{eq: threshold2}, provided that $C^{*1}_{t+1}>C^{*2}_{t+1}$. The latter condition is guaranteed by \cref{le: lastC} and induction using \cref{le: Cs}. Because this holds for any two agents, it also holds for the entire collection. The weakly ordered result follows by considering agents with equal deadlines. Because such devices are indistinguishable (other aspects were already identical), symmetry requires that their threshold bids are identical.

            	 \section{Proof of Lemma 7}
            	 \label{sec:pricediffproof}
            	 
            	 \begin{proof}
            	         First, if no devices are available to start, then $x_t=x_t^*$ by virtue of $P^*_t$ being optimal, and 0 lies within the bounds of \eqref{eq:pricebounds} as required. 
        
        Next, we prove the lower bound for the clearing price. The case $x_t < x_t^*$ occurs when agents have (unrealistically) low expectations for future costs and therefore submit low bids. We consider the lowest possible bid placed by a device that \emph{should} start at $t$ according to $\mathcal{S}$. \cref{th: diversity} guarantees that the lowest bid is placed by the device $a$ with the latest deadline $d^a$ (among devices that should start at $t$). The 
        bid $\hat{x}^{\downarrow}_t$ of this device is bounded by the lowest possible cost associated with waiting, according to the inequality:
        \begin{align}
         C^{s,a}_t(\hat{x}^{\downarrow}_t)& \ge \min_{t' \in \{ t+1,\ldots,d^a-D \}} \Delta t \sum_{i=0}^{D-1} P_i \left[ m_{t'+i}(P^*_{t'+i}) - \eta \right] \nonumber \\
            & \ge \min_{t' \in \{ t+1,\ldots,d^a-D \}} K^a_{t'} - \eta \Delta t  \sum_{i=0}^{D-1} P_i,
        \end{align}
        where the lower bounds of the price forecasts $X_{t'}$ and  \eqref{eq:Kdefinition} have been used. Cost-optimality of $\mathcal{S}$ and the fact that the device should start at $t$ implies $K_t \le K_{t'}$ \eqref{eq:Kcondition}, so that
        \begin{align}
          C^{s,a}_t(\hat{x}^{\downarrow}_t) & \ge K_{t}^a - \eta \Delta t  \sum_{i=0}^{D-1} P_i, \\
            & \ge \Delta t \sum_{i=0}^{D-1} P_i \left[  m_{t+i}(P^*_{t+i} - P_i) - \eta \right].
        \end{align}
        Expanding $C_t^{s,a}(\hat{x}^{\downarrow}_t)$ using \eqref{eq: startC} and making use of definition \eqref{eq:deltamdef} results in the inequality
        \begin{equation}
          P_0\left[\hat{x}^{\downarrow}_t - x_t^*\right] \ge - \sum_{i=0}^{D-1} P_i \left[ \Delta m_{t,i} + \eta \right]
        \end{equation}
        If $x_t<x_t^*$, at least one fewer device has started than was accounted for in the allocation $\mathcal{S}$. Therefore, the bid $\hat{x}^{\downarrow}_t$ of the marginal device must bound the clearing price $x_t$ from below, and the lower bound of \eqref{eq:pricebounds} follows. 
        
        The upper bound can be derived by considering devices that should \emph{not} start under $\mathcal{S}$, but submit a high bid because starting immediately appears to be cheaper than starting at their scheduled time $t^a > t$. The magnitude of such bids $\hat{x}^a_t$ is bounded by 
        \begin{align}
         C^{s,a}_t(\hat{x}^a_t) & = C^{\mathrm{*a}}_{t+1} \\ 
         & \le C^{s,a}_{t^a}(\overline{x}_{t^a}) \\
         & = \Delta t \sum_{i=0}^{D-1} P_i \left[ m_{t^a+i}(P^*_{t^a+i}) \right],
        \end{align}
        where we have used \eqref{eq:thresholdequality} and the fact that the optimal expected cost $C^{\mathrm{*a}}_{t+1}$ is always upper-bounded by the expected cost for one specific feasible time $t'>t$, including the special case $t'=t^a$. 
        Using the definition of $C^{s,a}_t$:
        \begin{align}
            \Delta t P_0 \left[\hat{x}^a_t - x_t^* \right] & \le  \Delta t \sum_{i=0}^{D-1} P_i \left[ m_{t^a+i}(P^*_{t^a+i}) - m_{t+i}(P^*_{t+i}) \right] \nonumber
            \\
            & \le \Delta t \sum_{i=0}^{D-1} P_i \left[ m_{t^a+i}(P^*_{t^a+i}) \right] - K_t^a \nonumber \\
            & \le \Delta t \sum_{i=0}^{D-1} P_i \left[ m_{t^a+i}(P^*_{t^a+i}) \right] - K_{t^a}^a, \nonumber \\
            & \le \Delta t \sum_{i=0}^{D-1} P_i \Delta m_{t^a, i}
        \end{align}
        where we have used the fact that $K_{t^a}^a \le K_t^a$ due to cost-optimality of $\mathcal{S}$.
        The bound for the highest bid  $\hat{x}^{\uparrow}_t$ of a device that should \emph{not} run is therefore defined as 
        \begin{align}
            P_0 \left[\hat{x}^{\uparrow}_t - x_t^* \right] & \le \max_{a \in \{a' \in \mathcal{A}: t^a > t\}} \sum_{i=0}^{D-1} P_i \Delta m_{t^a,i} \\
            & \le \max_{t' \in \{t,\ldots,  T\}} \sum_{i=0}^{D-1} P_i \Delta m_{t',i}
        \end{align}
        If $x_t>x_t^*$, at least one more device started than in the allocation $\mathcal{S}$, so the clearing price $x_t$ must be at or below $\hat{x}^{\uparrow}_t$, and the upper bound of \eqref{eq:pricebounds} follows.
        \end{proof}
            	 
	        \end{appendix}

\end{document}